\documentclass[submission,copyright,creativecommons]{eptcs}
\providecommand{\event}{NCMA 2023} 

\usepackage{amsmath}
\usepackage{iftex, amssymb}

\ifpdf
  \usepackage{underscore}         
  \usepackage[T1]{fontenc}        
\else
  \usepackage{breakurl}           
\fi
\newtheorem{definition}{Definition}

\newtheorem{example}{Example}

\newtheorem{proposition}{Proposition}
\newtheorem{remark}{Remark}

\newenvironment{proof}{\begin{trivlist}
                       \item[]{\bf Proof}
                       \hspace{0cm} }{\hfill {\large $\bullet$}
                       \end{trivlist}}

\newcommand{\be}{\begin{eqnarray}}
\newcommand{\ee}{\end{eqnarray}}

\newcommand{\LL}{{\cal L}}
\newcommand{\ignore}[1]{}

\newcommand{\der}{\rightarrow}

\newcommand{\derstarG}{\Longrightarrow  \hspace*{-5mm} {}^{{}^{*}}_{{}_{\hspace*{0.0mm}G}} \hspace*{4mm}}

\title{On Languages Generated by Signed Grammars}
\author{\"{O}mer E\u{g}ecio\u{g}lu
\institute{Department of Computer Science\\
University of California\\
Santa Barbara, CA 93106, USA}
\email{omer@cs.ucsb.edu}
\and
Benedek Nagy
\institute{Department of Mathematics,
Eastern Mediterranean University\\
99628 Famagusta, North Cyprus, Mersin-10, Turkey\\
Department of Computer Science, Institute of Mathematics and Informatics,\\ Eszterh\'azy K\'aroly Catholic University, Eger, Hungary}
\email{nbenedek.inf@gmail.com}
}

\def\titlerunning{Signed grammars}
\def\authorrunning{\"{O}mer E\u{g}ecio\u{g}lu \& Benedek Nagy}
\begin{document}
\maketitle

\begin{abstract}
We consider languages defined by signed grammars
which are similar to context-free grammars
except productions with signs associated to them are allowed. As a consequence, the words generated also
have signs.
We use the structure of the formal series of yields of all derivation trees over such a grammar as
a method of specifying a formal language and study properties of the resulting family of languages.
\end{abstract}

\section{Introduction}

We consider properties of signed grammars,  which are grammars
obtained from context-free grammars (CFGs) by allowing right hand sides of productions to have negative signs in front.
The concept of generation for such grammars is somewhat different from that of context-free grammars.
A signed grammar is said to generate a language $\LL$ if the formal sum of
the yields over all derivation trees over the grammar corresponds to the list of words in $\LL$.
For a signed grammar, the yields of derivation trees may have negative signs attached to them, but
the requirement is that when the arithmetic operations are carried out in the
formal sum, the only remaining words are those of $\LL$, each appearing with multiplicity one.

The structure of context-free languages (CFLs)
under a full commutation relation defined on the terminal alphabet is the central
principle behind Parikh's theorem \cite{Parikh61}.
In partial commutation, the order of letters of some pairs of the terminal alphabet is
immaterial, that is, if they appear consecutively, the word obtained by swapping their order
is equivalent to the original one. These equivalence classes are also called
traces and studied intensively in connection to parallel processes \cite{Mazurkiewicz77,JanickiKKM19,LATAtrace,DiekertRozenberg95}.
Our motivation for this work is languages obtained by picking
representatives of the equivalence classes in $ \Sigma^*$ under
a partial commutativity relation, called Cartier-Foata languages \cite{CartierFoata69}. In
the description of these languages with Kleene-closure type expansions, words appear with
negative signs attached to them. However
such words are cancelled by those with positive signs, leaving only the sum of the words of the language.
An example of this is $(a+b-ba)^*$ which is more familiarly denoted by the regular expression $a^*b^*$.
The interesting aspect of Cartier-Foata languages is that the words with negative signs cancel out
automatically, leaving only the
representative words, each appearing exactly once.

Motivated by these languages, we consider grammars which are obtained from context-free
grammars by
allowing signed productions, i.e., normal productions (in the role of positive productions) and productions
of the form $ A \der - \alpha$ (negative productions).
In this way, a derivation results in a signed word where
the sign depends on the parity of the number of negative rules applied in the derivation.
We consider those derivations equivalent that belong to the same derivation tree, and actually,
the derivation tree itself defines the sign of the derived word.
The language generated by such a grammar is obtained by
taking all possible derivation trees for a given word (both its positive and negative derivations)
and requiring that the
sum of the yields of all derivation trees over the grammar simply is a list of the
words in a language $\LL$. This means that the simplified formal sum
is of the form $ \sum_{w \in \LL} w $, each word of the language appearing with multiplicity one. (Without loss of generality, in this study, we restrict ourselves to grammars having finitely many parse trees for each of the derived words.)

On one hand, the requirements in the specification of a language generated by a
signed grammar may seem too restrictive.
But at the same time this class of languages includes all unambiguous context-free
languages and it is closed under complementation, and consequently can
generate languages that are not even context-free. Therefore it is of interest to consider the
interplay between the restrictions and various properties of languages generated by signed grammars.

\section{Preliminaries}
Given a language $\LL$ over an alphabet $\Sigma$, we identify $\LL$ with the formal
sum of its words denoted by $f(\LL)$:
\begin{equation}\label{fsum}
f(\LL)= \sum_{w \in \LL} w ~ .
\end{equation}
The sum in (\ref{fsum}) is also referred to as the {\em listing series} of $\LL$.
A {\em weighted series of $\LL$} is a formal series of the form
$ \sum_{w \in \LL} n_w \, w $
where $n_w$ are integers. Thus a weighted series of $ \Sigma^*$
$$ \sum_{w \in \Sigma^*} n_w \, w $$
is the listing series of some language $\LL$ over $\Sigma$ iff
\begin{equation}\label{defnw}
n_w =
\left\{
\begin{array}{ll}
1 & \mbox{ if } w \in \LL \\
0 & \mbox{ if } w \not\in \LL \,.
\end{array}
\right.
\end{equation}
We are allowed ordinary arithmetic operations on weighted series in a natural way.
The important thing is that a weighted series is the listing series of a language
$\LL$ iff the coefficients of the words in $\LL$ in the weighted series are 1, and all the others are 0.
So for example over $\Sigma=\{ a,b,c\}$, the weighted series
$ a+b+c+ba $
is the listing series of the finite language $ \LL=\{a,b,c,ba\}$, whereas
the weighted series $ a+b+c-ba $
does not correspond to a language over $\Sigma$. This is because in the latter example $n_w$
does not satisfy (\ref{defnw}) for $ w = ba$. As another example, the difference of the
weighted series $ 2 a + 3 b -c +ba$ and $a+2b-2c+ba$
corresponds to the language $\LL= \{a,b, c \}$. \\

\subsection{CFGs and degree of ambiguity}
Next we look at the usual CFGs $G = (V,\Sigma, P , S)$. Here the start symbol is $ S\in V$.
Let $T$ be a parse (derivation) tree over $G$ with root label $S$ and terminal letters as labels of the leaves of $T$.
Let
$ Y (T) \in \Sigma^* $ be the {\em yield} of $T$. Then
the language generated by $G$ is
$$
\LL(G)= \{ Y (T) ~|~ T \mbox{ is a parse tree over } G \} \,.
$$
This is equivalent to $ \LL(G)= \{ w \in \Sigma^* ~|~ S \derstarG w \} $.
For a CFG $G$, we can define the formal weighted sum
\begin{equation}
\label{weighted}
f(G) = \sum_{T \in {\cal T}_G } Y(T) =\sum_{w \in \Sigma^*} n_w w
\end{equation}
where ${\cal T}_G$ denotes all parse trees over $G$.
Various notions of ambiguity for CFLs can be interpreted as
the nature of the coefficients $n_w$ that appear in (\ref{weighted}).
Rewriting some of the definitions in Harrison
\cite[pp. 240-242]{Harrison78} in terms of these coefficients, we have
\begin{enumerate}
\item Given $k\geq1$, $G$ is {\em ambiguous of degree $k$} if $n_w \leq k$ for all $ w \in \LL(G)$.
\item ${\LL}$ is {\em inherently ambiguous of degree $k\geq2$} if $\LL$ cannot be
generated by any grammar that is ambiguous of degree less than
$k$ but can be generated by by a grammar that is ambiguous of degree $k$.
In other words the degree of ambiguity of a CFL is the
least upper bound for the number of derivation trees which a word in the language can have.
\item $\LL$ is {\em finitely inherently ambiguous} if there is some $k$ and some $G$ for $\LL$ so that
$G$ is inherently ambiguous of degree $k$.
\item A CFG $G$ is {\em infinitely ambiguous} if for each $i \geq 1$, there exists a word in $\LL(G)$
which has at least $i$ parse trees. A language $L$ is {\em infinitely inherently ambiguous}
if every grammar generating $L$ is infinitely ambiguous. \\
\end{enumerate}

The CFL ${\cal A} = \{ a^i b^j c^k ~ | ~  i=j \mbox{ or } j=k \}$
is inherently ambiguous of degree 2 \cite[p. 240]{Harrison78},
$ {\cal A}^m $ is inherently ambiguous of degree $2^m$ \cite[Theorem 7.3.1]{Harrison78},
and $ {\cal A}^* $ is  infinitely inherently ambiguous  \cite[Theorem 7.3.3]{Harrison78}.
Another interesting CFL which is infinitely inherently ambiguous is Crestin's language \cite{Crestin72}
of double palindromes over a binary alphabet
$\{ w_1 w_2 ~ | ~ w_1, w_2 \in \{a,b\}^*, w_1 = w_1^R, w_2 = w_2^R \}$.
Furthermore, for every $ k \geq 1$, there exist
inherently ambiguous CFLs of degree $k$.
The behavior of the sequence $n_w$ over all CFGs for a language was studied by Wich \cite{Wich00, Wich05}.

\section{Signed grammars}
We consider {\em signed grammars} $G$ which are like CFGs 
but with a sign associated with each production, that is, apart from the usual (say positive) productions,
we allow productions of the form $ A \der - \alpha$. 
In the derivation relation we use the signs as usual in a multiplicative manner:
We start the derivation from the sentence symbol (with $+$ sign, but as
usual we may not need to put it, as it is the default sign).
The derivation steps, as rewriting steps, occur as they are expected in a CFG, the only
extension is that we need to deal with also the sign.
When a positive  production is applied in a sentential form, its sign does not change,
while whenever a negative production is applied,
this derivation step switches the sign of the sentential form.
Thus, 
in this case the yield of a parse tree of
$G$ is a word over $ \Sigma$ with a $ \pm$ sign attached to it.
Furthermore, the sign of a derived word depends only on the parity of the
number of
negative productions used during its derivation. Therefore, different derivation trees for the same
word may lead to the word with different signs attached to it. 
We note that, in fact, any CFG is a signed grammar.
For a signed grammar $G$, let
$f(G)$ be defined as in (\ref{weighted}),
where again ${\cal T}_G$ denotes all parse trees over $G$. Without loss of generality, we may assume that in the grammar $G$ there are only finitely many parse trees for any of the words generated by the grammar.

\begin{definition}
We say that a signed grammar $G$ {\em generates} a language $\LL$ iff the weighted series $f(G)$
in (\ref{weighted}) is the listing series of $\LL$, i.e. $f(G) = f(\LL)$.
\end{definition}

\subsection{Examples of languages generated by signed grammars}
\begin{example}
\normalfont
For the signed grammar $G_1$ with start symbol $A$ and productions $ A \der -a A \, | \, \lambda $, we have
\begin{equation}\label{G1}
f(G_1) = \sum_{i\geq 0} a^{2i} -
 \sum_{i\geq 0} a^{2i+1} \,.
\end{equation}
Therefore the signed grammar $G$ with productions
$S \der A \,| \, B$, $A \der -a A \, | \, \lambda$, $ B \der aa B \,  | \, a $
generates the regular language $(aa)^*$.
As this is our first example, we provide details of the derivations in $G$: \\ 
\begin{itemize}
\item The empty word $\lambda$ can be derived only in one way, by applying a positive production, thus it is in the language.
\item By applying a negative and a positive production, $S \Rightarrow A \Rightarrow - aA \Rightarrow -a$ yields
$-a$, and $S \Rightarrow B \Rightarrow a$ yields $+a$. These two are the only derivations over $G$ for
$\pm a$. This means that the word $a$ is not in the language.
\item For the word $aa$, the only derivation is
$S \Rightarrow A \Rightarrow - aA \Rightarrow aaA \Rightarrow aa $.
Consequently $aa$ is in the generated language.
\item Finally, by induction, one can see that an even number of $a$-s can only be produced by starting the
derivation by $ S \Rightarrow A$. Following this positive production, each usage of $ A \der -aA$ introduces a
negative sign. Therefore each word of the form $ a^{2i}$ is generated once this way with a $+$ sign.
On the other hand
there are two possible ways to produce a string $ a^{2i+1}$ of an odd number of $a$-s.
One of these starts with $ A \Rightarrow -aA$ as before and produces $ - a ^{2i+1}$ after an odd number
of usages of $ A \der -aA$;
the other one starts with $ S \Rightarrow B$ and produces
$  a ^{2i+1}$ after an even number of applications of $ B \der aa B$, followed by $B \der a$.
Therefore odd length words cancel each other out and are not in the language generated.
\end{itemize}
Another way to look at this is to note that
for the (signed) grammar $G_2$ with the start symbol $B$ and productions
$ B \der aa B \,  | \, a $, we have
\begin{equation}\label{G2}
f(G_2) = \sum_{i\geq 0} a^{2i+1} \,,
\end{equation}
and the words generated by $G$ are given by the formal sum of (\ref{G1}) and (\ref{G2}).
\end{example}
\begin{example}
\normalfont
The signed grammar with productions $ S \der aS \,|\, bS \,|\, - ba S \,|\, \lambda$
generates the regular language denoted by the regular expression $ a^* b ^*$.
First few applications of the productions give
\begin{eqnarray*}
&& \lambda; \\
&& a + b -ba;\\
&& a^2 + ab -aba+ba+b^2 -b^2a -ba^2-bab+baba;
\end{eqnarray*}
in which the only immediate cancellation is of $-ba$, though all words carrying negative signs will eventually cancel out.
This is a special case of the Cartier-Foata result \cite{CartierFoata69}, \cite[Section 8.4]{EgeciogluGarsia21}.
\end{example}
\begin{example}
\normalfont
Over the decimal (or the binary) alphabet 
 we can construct an
unambiguous regular grammar $G$
that generates all nonnegative even numbers, e.g., 
$S \der 9S \,|\, 8A \,|\, 7S \,|\, 6A \,|\, 5S \,|\, 4A  \,|\, 3S \,|\, 2A \,|\, 1S \,|\, 0A $ and $ A \der 9S \,|\, 8A \,|\, 7S \,|\, 6A \,|\, 5S \,|\, 4A  \,|\, 3S \,|\, 2A \,|\, 1S \,|\, 0A \,|\, \lambda  $.
Let, further, a regular grammar $G^{\prime}$ be generating the numbers which are
divisible by 6 (e.g., based on the deterministic finite automaton checking the sum of the digits to be divisible by 3 and the last digit must be even, we
need states/nonterminals to count the sum of already read digits by mod 3 and take care to the last digit as we
did for $G$).

Then $\LL(G)$ consists of all even numbers and $\LL(G^{\prime})$ consists of all numbers divisible by $6$.
Now, from $G^{\prime}$, we may make a signed grammar $G^{\prime\prime}$ which allows us to derive
every multiple of 6 with the sign $-$.
Then by combining the two grammars $G$ and $G^{\prime\prime}$, we can easily give a signed
grammar that generates all even numbers that are not divisible
by 3 (i.e., even numbers not divisible by 6).
\end{example}
\begin{example}
\normalfont
Over the  alphabet $\{a,b\}$ consider the signed grammar with productions
$ S \der aSa \, | \, bSb \, | \, a \, | \, b$. This so far generates odd length palindromes. Let us add
the productions $ S \der -A$, $~A \der  -abAba \, | \, a $.

Then each odd length palindrome with the letter $b$ in the middle
has exactly one derivation tree with a $+$ sign. There are no cancellations for these and
therefore all odd length palindromes with $b$ in the middle are in the language.
If the middle of an odd length palindrome $w$ is $a$ but not $ababa$,
then $w$ is not in $\LL$ as it has also derivation tree with $-$ sign.
Similarly, if the middle of $w$ is $ababa$ but not $ababababa$,
$w$ is in $\LL$. In general, if an odd length palindrome $w$ has
$(ab)^{2k-1}a(ba)^{2k-1}$ in the middle, but it does not have
$(ab)^{2k}a(ba)^{2k}$ in its middle, then it is in $\LL$.
Here the number of derivation trees
for a word with a $+$ sign is either equal to the number of derivation
trees with a $-$ sign for the word, or it is exactly one more.
\end{example}
\begin{example}
\normalfont
For the following signed grammar
\begin{eqnarray*}
&& S_1 \der - aA \, | \, Ba \, | \, a \\
&& A \der - aA \, | \, Ba\,  |\,  a \\
&& B \der - aB \,|\, Ba \,| \,-a\, |\, a a
\end{eqnarray*}
for $n$ odd, there are $2^{n-1}$ parse trees for $a^n$ and $2^{n-1} - 1$ parse trees for $ - a^n$.
For $n$ even,
there are $2^{n-1} - 1$ parse trees for  $a^n$ and $2^{n-1}$  parse trees for $- a^n$. In other words for the above grammar
\begin{eqnarray*}
f(G) & = &\sum_{i \geq 0} 2^{2i} a^{2i+1} +  \sum_{i \geq 0} ( 2^{2i}-1) a^{2i}
- \sum_{i \geq 0} (2^{2i}-1) a^{2i+1} -  \sum_{i \geq 0} 2^{2i} a^{2i}\\
&=&
\sum_{i \geq 0} (-1)^i a^{i+1} \,.
\end{eqnarray*}

If we add the productions
$ S \der S_1 \, | \, S_2, ~ S_2 \der aaS_2 \, | \, aa$ then the resulting signed grammar generates the regular language
$a (aa)^*$.
Even though the language generated is very simple we see that signed grammars possess some interesting
behavior.
\end{example}

\section{Properties of languages generated by signed grammars}

In this section our aim is twofold. On the one hand we give some closure properties
of the class of languages generated by our new approach and, on the other hand, we give
hierarchy like results by establishing where this family of languages is compared to
various other classes.

We immediately observe that
in the weighted sum (\ref{weighted}) for a CFG $G$ (i.e. a signed grammar $G$ with no
signed productions), the coefficient $n_w$
is the number of parse trees for $w$ over $G$, in other words the degree of ambiguity of $w$.

\begin{proposition}
Any unambiguous CFL is generated by a signed grammar.
\end{proposition}
\begin{proof}
An unambiguous CFL $\LL$ is generated by the signed grammar $G$ where
$G$ is any unambiguous CFG for $\LL$.
\end{proof}

As the class of unambiguous CFLs contains all deterministic CFLs, $LR(0)$ languages,
regular languages, subsets of $w_1^* w_2^*$ \cite[Theorem 7.1]{GinsburgUllian66},
all of these languages are generated by signed grammars.
Further, all these classes are proper subsets of the class of languages generated by signed grammars.

Now we present a closure property.

\begin{proposition}
Languages generated by signed grammars are closed under complementation.
\end{proposition}
\begin{proof}
Take an unambiguous CFG for $\Sigma^*$ with start symbol $S_1$.
If $\LL$ is generated by a signed grammar with start symbol $S_2$ (and no common nonterminal in the two grammars),
then the productions of the two grammars together with $ S \der S_1 ~|~ - S_2$
with a new start symbol $S$ generates $\overline{\LL}$.
\end{proof}

We continue the section comparing our new class of languages with other well-known language class, the class of
CFLs.

In 1966 Hibbard and Ullian constructed an unambiguous CFL
whose complement is not a CFL \cite[Theorem 2]{HibbardUllian66}.
Recently Martynova and Okhotin constructed
an unambiguous linear language whose complement is not context-free \cite{MartynovaOkhotin23}.
This shows that unambiguous linear CFLs are not closed under complementation
while providing another proof of Hibbard and Ullian's result.

We know that languages generated by signed grammars are closed under complementation,
and also every unambiguous CFL is generated by a signed grammar. A consequence of
this is that signed grammars can generate languages that are not context-free.

\begin{proposition}\label{NotCFL}
There is a language generated by a signed grammar that is not context-free.
\end{proposition}
\begin{proof}
If $\LL$ is the unambiguous CFL constructed by Hibbard and Ullian, then $\LL$ and
therefore $\overline{\LL}$ are generated by signed grammars. But we know that $\overline{\LL}$ is not context-free.
\end{proof}

Actually, our last proposition shows that the generative power of signed grammars is surprisingly large, it contains, e.g.,
all deterministic and unambiguous 
CFLs and their complements. Thus, one can easily generate some languages that are not in the class of CFLs.

Continuing with closure properties,
recall that disjoint union is an operation that is defined only on disjoint sets
which produces their union.


\begin{proposition}
Languages generated by signed grammars are closed under disjoint union $\uplus$.
\end{proposition}
\begin{proof}
Let $\LL_1$ and $\LL_2$ be two languages over an alphabet $\Sigma$ such that $\LL_1 \cap \LL_2 = \emptyset$.
Let $\LL_1$ be generated by a signed grammar with start symbol $S_1$ and
$\LL_2$ be generated by a signed grammar with start symbol $S_2$, such that the sets of nonterminals of these two grammars are disjoint.
Then the productions of the two grammars together with $ S \der S_1 ~|~ S_2$
with a new start symbol $S$ generates the disjoint union $\LL_1 \uplus \LL_2$.
\end{proof}

Now, let us define the set theoretical operation ``subset minus'' ($\ominus$), as follows:
let $A\subseteq B$, then $B \ominus A = B \setminus A$.
This type of setminus operation is defined only for sets where the subset condition holds.

\begin{proposition}
Languages generated by signed grammars are closed under subset minus $\ominus$.
\end{proposition}
\begin{proof}
Let $\LL_1 \subseteq \LL_2$ be two languages over a given alphabet $\Sigma$.
Take the signed grammar for $\LL_1$ with start symbol $S_1$.
If $\LL_2$ is generated by a signed grammar with start symbol $S_2$ (with no common nonterminals of the two grammars),
then the productions of the two grammars together with $ S \der S_1 ~|~ - S_2$
with a new start symbol $S$ generates the language of $\LL_2 \ominus \LL_1$.
\end{proof}

Let $\LL_1,\LL_2 \subseteq \Sigma^*$ be two languages and $\$ \not\in \Sigma $.
The $\$$-concatenation of $\LL_1$ and $\LL_2$ is the language $\LL_1 \$ \LL_2 $
over the alphabet $ \Sigma \cup \{\$\}$.

\begin{proposition}
Languages generated by signed grammars are closed under $\$$-concatenation.
\end{proposition}
\begin{proof}
The language $\LL_1 \$$
has the prefix property (i.e. it is prefix-free) due to the special role of the marker $\$$.
Let $G_1$ and $G_3$ be signed grammars with disjoint variables and start symbols $S_1$ and $S_3$
that generate $\LL_1$ and $\LL_2$, respectively. Consider also the signed grammar $G_2$ with the single production
$S_2 \der \$$. Then the signed grammar which have all the productions of $G_1, G_2, G_3$  together with
the production $ S \der S_1 S_2 S_3 $ where $ S$ is a new start symbol generates
the language $\LL_1 \$ \LL_2 $.  The proof follows by observing that for $u, u' \in \LL_1$ and $v, v' \in \LL_2$,
$u \$ v = u' \$ v'$ iff $u=u'$ and $v=v'$, so that each word that appears in the expansion of
$$
\left( \sum_{w \in \LL_1} w \right) \$ \left( \sum_{w \in \LL_2} w \right)
$$
has coefficient 1.
\end{proof}

In a similar manner, it can also be seen that we have a similar statement
for languages over disjoint alphabet, i.e., the class of languages generated by signed grammars is
closed under ``disjoint concatenation'' $\boxdot$.

\begin{proposition}
Let $\LL_1 \subseteq \Sigma_1^*$ and $\LL_2 \subseteq \Sigma_2^*$ be two languages that are generated by signed grammars, where $\Sigma_1 \cap \Sigma_2 = \emptyset$. Then, the
language $\LL_1 \boxdot \LL_2 = \LL_1 \LL_2$ can be generated by a signed grammar.
\end{proposition}

In the following proposition,
$f(\LL)$ and  $f(G)$ are as defined in (\ref{fsum}) and (\ref{weighted}).
\begin{proposition}\label{diff}
Suppose  $\LL$ generated by a signed grammar. Then there are CFGs $G_1$ and $G_2$ such that
$f(\LL) = f(G_1) - f(G_2)$.
\end{proposition}
\begin{proof}
Given a signed grammar over $\Sigma$, add an extra letter $t$ to
$\Sigma$ and replace all productions of the form $ A \der - \alpha $ by  $ A \der t \alpha$.
The words generated by this CFG over $ \Sigma \cup \{t\}$ with an even number of occurrences of
$t$ is a CFL since it is the intersection of CFL and the regular language, i.e. all words over
$ \Sigma \cup \{t\}$ with an even number of occurrences of $t$.
Similarly, the words generated with an odd number of occurrences of $t$ is a CFL.
We can then take homomorphic images of these two languages generated by replacing $ t$ by $ \lambda$
and obtain two CFLs generated by CFGs $G_1$ and $G_2$.
The weighted series $f(G)$ is then the difference of two weighted series
\begin{equation}\label{nw}
f(G)= f(G_1) - f(G_2) = \sum_{w \in \Sigma^*} n_w w ~- ~ \sum_{w \in \Sigma^*} n'_w w  ~.
\end{equation}
In (\ref{nw}), the coefficients $n_w$ and $n'_w$ are nonnegative integers
for all $w\in \Sigma^*$ as they count the number of derivation trees for $w $ over $G_1$ and $G_2$, respectively.
\end{proof}

\begin{remark}
In Proposition \ref{diff}, $f(G_1) - f(G_2)$ is the listing series of $\LL$,
and therefore $ n_w - n'_w =1 $ or $ n_w - n'_w =0 $ for all $ w \in \Sigma^*$.
In the first case $ w\in \LL$, and in the second $ w \not\in \LL$.
Note that these conditions do not imply that $ \LL = \LL(G_1) \setminus \LL(G_2)$.
\end{remark}

\section{Partial commutativity}
Addition of commutativity relations to CFGs was considered in \cite{Benedek09}.
Here we consider partial commutativity defined on $\Sigma^*$ where
$\Sigma =\{x_1, x_2, \ldots, x_m\}$.
Given an $m \times m $ symmetric $\{0 , 1\}$-matrix $ A = [ a_{i,j}] $ with 1s down the diagonal, a pair of letters
$ x_i, x_j $ is a commuting pair iff $ a_{i,j}  = 1 $.
This defines an equivalence relation and partitions $ \Sigma^*$ into equivalence classes, 
also known as traces. Thinking about the element of the alphabet as processes and traces as their scheduling,
commuting processes are considered as independent from each other. In this way the
theory of traces has been intensively studied in connection to parallel processes \cite{JanickiKKM17,JanickiKKM19}.
A (linearization of a) trace language is a union of some of these equivalence classes. Trace languages based on regular, linear and context-free languages (adding a partial commutativity relation to the language) were studied and accepted by various types of automata with translucent letters in \cite{LATAtrace,LINtrace,CFtrace}, respectively.
   Traces and trajectories are also analyzed in various grids \cite{MateescuRS98,MateescuSY00,IWCIA17}.
On the other hand, 
the {\em Cartier--Foata language} $\LL_A$
corresponding to the matrix $A$ of a partial commutativity relation is constructed by picking a representative word from each equivalence class.

Let us define a set $ F \subseteq \Sigma $ to be {\em commuting} if any pair of
letters in $ F $ commute.  Let $ {\cal C}(A) $ denote the collection of all nonempty commuting sets.
Denote by $w(F)$ the word obtained by juxtaposing the letters of $ F $.  The order in which
these letters are juxtaposed is immaterial since all arrangements are equivalent.

The central result is that the listing series $f(\LL_A)$  can be constructed
directly from the matrix $ A $:
\begin{equation}\label{CFexpansion}
f({\cal L}_{A}  )   =  \left( \sum_{ F \in {\cal C}(A) } (-1)^{\# F}  w (F) \right)^*
= \sum_{n \geq0} \left( \sum_{ F \in {\cal C}(A)} (-1)^{\# F}  w (F) \right)^n ~ ,
\end{equation}
where $\#F$ denotes the number of elements of $F$.

Over $ \Sigma=\{a,b\}$ where $a$ and $b$ commute, the Cartier-Foata theorem gives $\LL_A$ as
$(a+b - ba)^*$, which is to be interpreted as the weighted series
$ \lambda + (a+b-ba) + (a+b-ba)^2 + \cdots $ In this case the representatives of the equivalence classes
are seen to be the words in $ a^* b^*$. The essence of the theorem is that this is a listing series,
so there is exactly one representative word from each equivalence class that remains after
algebraic cancellations are carried out.

Similarly
over $\Sigma = \{ a, b ,c \}$ with $ a, b$ and $a,c$ commuting pairs, the listing series is
$ \lambda + (a+b+c-ba-ca) + (a+b+c-ba-ca)^2 + \cdots $



The words in this second language are generated by the signed grammar
$$
S \der \lambda \,|\, a S \,|\, bS \,|\, cS \,|\, -ba S \,|\, - ca S ~.
$$

\section{Conclusions and a conjecture}
Proposition \ref{diff} provides an expression for the
listing series of a language generated by a signed grammar in terms of weighted listed series of
two CFLs. However this result is short of a characterization in terms of CFLs.
It is also possible to change the way signed grammars generate languages by requiring $n_w \geq 1$ in
(\ref{defnw}) instead of equality. 
 In this way, every signed grammar would generate a language, and obviously, the class of generated languages would also change.
However, our consideration in this paper to allow only $0$ and $1$ to be the signed sum, gives a nice and immediate connection to Cartier-Foata languages in the regular case by special regular like expressions.

Since by signed grammars, we generate languages based on counting the number of (signed) derivation
trees, it is
straightforward to see the connection between our grammars and unambiguous CFLs.
On the other hand, there may be
more than one derivation tree for a given word $w$,
with the proviso that the algebraic sum of the yields of derivation trees
for it
has multiplicity  $n_w \in \{0,1\}$. Therefore signed grammars may
also generate ambiguous CFLs. In this sense, the bottom
of the hierarchy, the unambiguous CFLs are included in the class we have investigated.
On the other hand, if there are multiple derivation
trees for a word generated by a grammar, by playing with their signs, we have
a chance to somehow have their signed sum to be in $\{0,1\}$. Thus, it may be
possible to
generate languages that are higher in the hierarchy based on ambiguity.
However, this is still an open problem.

We have shown that signed grammars can generate languages that are
not context-free.
It would be of interest to use the fact that the languages generated by signed grammars are closed
under complementation  to show that signed grammars can generate inherently ambiguous CFLs.
One way to do this
would be to start with an unambiguous CFL whose complement is an inherently ambiguous CFL.
The standard examples of inherently ambiguous CFLs do not seem to have this property.
By the Chomsky-Sch\"{u}tzenberger
theorem \cite{ChomskySchutzenberger63} the generating function of an unambiguous CFL is algebraic.
Using the contrapositive and analytical methods, Flajolet \cite{Flajolet87} and later Koechlin \cite{Koechlin22}
devised ingenious methods to show the transcendence of the generating function of a given language to prove its
inherent ambiguity.
However if the generating function of $\LL$ is transcendental so is the generating function of its
complement $\overline{\LL}$. This means that one needs to look among inherently ambiguous
languages with algebraic generating functions
(e.g. $\{a^i b^j c^k ~|~ i=j ~\mbox{or}~ j = k \}$, see \cite[Proposition 14]{Koechlin22})
if the complement has any chance of being unambiguous.

So it would be nice to have an answer to the following question:
{\em Is there an unambiguous CFL whose complement is an inherently ambiguous CFL?}

A related problem of showing the existence of an inherently ambiguous CFL
whose complement is also an inherently ambiguous CFL was settled by Maurer \cite{Maurer70}.

\nocite{*}
\bibliographystyle{eptcs}
\bibliography{signed_grammars_bibliography}
\end{document}